\documentclass[conference]{IEEEtran}
\usepackage[letterpaper, left=0.91in, right=0.91in, bottom=0.91in, top=0.75in]{geometry}
\newlength \figwidth
\usepackage{cite}
\ifCLASSINFOpdf
  \usepackage[pdftex]{graphicx}
	\usepackage{epstopdf}
  \graphicspath{{../Figures/}}
  \DeclareGraphicsExtensions{.eps}
\else
  \usepackage[dvips,draft]{graphicx}
\fi
\usepackage[cmex10]{amsmath}
\usepackage{amssymb}
\usepackage{pifont}
\usepackage{amsthm}
\usepackage{url}
\usepackage{dsfont}
\usepackage{tabulary}
\usepackage{multirow}
\usepackage{color}

\usepackage{color}
\usepackage{times}
\usepackage{verbatim}
\usepackage{floatflt}
\usepackage{enumerate}
\usepackage{array}
\usepackage{multicol,afterpage,wrapfig} 
\usepackage{tikz}
\usepackage[noend]{algpseudocode}
\makeatletter
\def\BState{\State\hskip-\ALG@thistlm}
\makeatother
\usepackage{amsmath,amssymb,eucal}
\usepackage[utf8]{inputenc}
\usepackage{epsfig}
\usepackage{exscale}
\usepackage{latexsym}
\usepackage{verbatim}
\usepackage{amsfonts}
\usepackage{multirow}
\usepackage{cite}
\usepackage{balance}
\usepackage{color}
\usepackage{transparent}
\usepackage{import}
\usepackage{mathtools}
\usepackage{tikz}
\usetikzlibrary{automata,arrows,positioning,calc}
\usepackage{bbm}
\usepackage[printonlyused,withpage]{acronym}
\usepackage{tabu}
\usepackage{longtable}
\usepackage{balance}
\usepackage{footnote}
\usepackage{tablefootnote}
\usepackage{enumitem}

\usepackage{adjustbox}
\usepackage{multirow}
\usepackage{pifont}
\def\BibTeX{{\rm B\kern-.05em{\sc i\kern-.025em b}\kern-.08em
    T\kern-.1667em\lower.7ex\hbox{E}\kern-.125emX}}
    
\renewcommand{\IEEEQED}{\IEEEQEDopen}

\newcommand*\xbar[1]{%
  \hbox{%
    \vbox{%
      \hrule height 0.5pt 
      \kern0.36ex
      \hbox{%
        \kern-0.12em
        \ensuremath{#1}%
        \kern-0.12em
      }%
    }%
  }%
}

\newfont{\bbb}{msbm10 scaled 500}

\newfont{\bb}{msbm10 scaled 1100}

\newcommand{\Rc}{{\cal R}}

\newcommand{\bx}{{\text{b}}}

\newcommand{\hx}{{\text{h}}}

\newcommand{\nx}{{\text{n}}}
\newcommand{\ox}{{\text{o}}}
\newcommand{\px}{{\text{p}}}

\newcommand{\sx}{{\text{s}}}

\renewcommand{\d}{\mathrm{d}}

\def\e{\text{e}}

\newcommand{\executeiffilenewer}[3]{%
\ifnum\pdfstrcmp{\pdffilemoddate{#1}}%
{\pdffilemoddate{#2}}>0%
{\immediate\write18{#3}}\fi%
}

\usetikzlibrary{arrows,calc}
\allowdisplaybreaks

\IEEEoverridecommandlockouts
\begin{document}
\pagenumbering{gobble}

\def \pcov{\mathcal{P}_\mathrm{cov}}
\def \pcovl{\mathcal{P}_\mathrm{cov}^\mathrm{L}}
\def \pcovn{\mathcal{P}_\mathrm{cov}^\mathrm{N}}
\def \pcovu{\mathcal{P}_\mathrm{cov}^\upsilon}
\def \pcovuav{\mathcal{P}_\mathrm{u}}
\def \pcovc{\mathcal{P}_\mathrm{g}}
\def \sinr{\mathsf{SINR}}
\def \du{d_\mathrm{u}}
\def \dc{d_\mathrm{g}}
\def \db{d_\mathrm{b}}
\def \ru{r_\mathrm{u}}
\def \Ru{R_\mathrm{u}}
\def \rc{r_\mathrm{g}}
\def \rb{r_\mathrm{b}}
\def \al{\alpha_\mathrm{L}}
\def \an{\alpha_\mathrm{N}}
\def \au{\alpha_\upsilon}
\def \bl{\beta_\mathrm{L}}
\def \bn{\beta_\mathrm{N}}
\def \bu{\beta_\upsilon}
\def \bx{\beta_\xi}
\def \pu{\mathrm{P_u}}
\def \pum{\mathrm{P_u^M}}
\def \pcm{\mathrm{P_g^M}}
\def \pc{\mathrm{P_g}}
\def \pcl{\mathrm{P_g^L}}
\def \pcn{\mathrm{P_g^N}}
\def \hu{\mathrm{h_u}}
\def \htx{h_\mathrm{t}}
\def \hrx{h_\mathrm{r}}
\def \hb{\mathrm{h_b}}
\def \hc{\mathrm{h_g}}
\def \pl{\mathcal{P}_\mathrm{L}}
\def \pn{\mathcal{P}_\mathrm{N}}
\def \pup{\mathcal{P}_\upsilon}
\def \px{\mathcal{P}_\xi}
\def \pluu{\mathcal{P}_\mathrm{L}^\mathrm{uu}}
\def \ml{m_\mathrm{L}}
\def \mn{m_\mathrm{N}}
\def \m{m_\upsilon}
\def \muu{\mathrm{m_{uu}}}
\def \muul{\mathrm{m_{uu}^L}}
\def \muun{\mathrm{m_{uu}^N}}
\def \mcb{\mathrm{m_{gb}}}
\def \mcbl{\mathrm{m_{gb}^L}}
\def \mcbn{\mathrm{m_{gb}^N}}
\def \mub{\mathrm{m_{ub}}}
\def \mubl{\mathrm{m_{ub}^L}}
\def \mubn{\mathrm{m_{ub}^N}}
\def \mcul{\mathrm{m_{gu}^L}}
\def \mcun{\mathrm{m_{gu}^N}}
\def \nl{n_\mathrm{L}}
\def \nn{n_\mathrm{N}}
\def \nx{n_\xi}
\def \n{n_\upsilon}
\def \sl{\psi_\mathrm{L}}
\def \sn{\psi_\mathrm{N}}
\def \su{\psi_\upsilon}
\def \sx{\psi_\xi}
\def \ol{\omega_\mathrm{L}}
\def \on{\omega_\mathrm{N}}
\def \ou{\omega_\upsilon}
\def \ox{\omega_\xi}
\def \il{I_\mathrm{L}}
\def \iN{I_\mathrm{N}}
\def \yl{y_\mathrm{L}}
\def \yn{y_\mathrm{N}}
\def \yu{y_\upsilon}
\def \lapi{\mathcal{L}_I}
\def \lapil{\mathcal{L}_{I_\mathrm{L}}}
\def \lapin{\mathcal{L}_{I_\mathrm{N}}}
\def \lapix{\mathcal{L}_{I_\xi}}
\def \lapic{\mathcal{L}_{I_\mathrm{g}}}
\def \lapiu{\mathcal{L}_{I_\mathrm{u}}}
\def \Ul{\Upsilon_\mathrm{L}}
\def \Un{\Upsilon_\mathrm{N}}
\def \Ux{\Upsilon_\xi}
\def \Uxu{\Upsilon_\xi^{\,\upsilon}}
\def \Unu{\Upsilon_\mathrm{N}^\upsilon}
\def \frl{f_{R_\mathrm{b}}^\mathrm{L}}
\def \frn{f_{R_\mathrm{b}}^\mathrm{N}}
\def \frx{f_{R_\mathrm{b}}^{\,\xi}}
\def \laml{\lambda_\mathrm{L}}
\def \lamn{\lambda_\mathrm{N}}

\def \t{\mathrm{T}}
\def \i{I}
\def \ic{I_\mathrm{g}}
\def \icl{I_\mathrm{g}^\mathrm{L}}
\def \icn{I_\mathrm{g}^\mathrm{N}}
\def \iu{I_\mathrm{u}}
\def \iul{I_\mathrm{u}^\mathrm{L}}
\def \iun{I_\mathrm{u}^\mathrm{N}}
\def \ful{\Phi_\mathrm{u}^\mathrm{L}}
\def \fun{\Phi_\mathrm{u}^\mathrm{N}}
\def \fu{\Phi_\mathrm{u}}
\def \fcl{\Phi_\mathrm{g}^\mathrm{L}}
\def \fcN{\Phi_\mathrm{g}^\mathrm{N}}
\def \fc{\Phi_\mathrm{g}}

\def \gb{\mathrm{G_b}}
\def \gu{\mathrm{G_{u}}}
\def \gb{\mathrm{G_b}}
\def \tt{\theta_\mathrm{t}}
\def \tb{\theta_\mathrm{b}}

\def \lamci{\hat{\lambda}_\mathrm{g}}
\def \lamb{\lambda_\mathrm{b}}
\def \lamu{\lambda_\mathrm{u}}
\def \lamuhat{\hat{\lambda}_\mathrm{u}}
\def \lamg{\lambda_\mathrm{g}}
\def \lamul{\lambda_\mathrm{u}^\mathrm{L}}
\def \lamun{\lambda_\mathrm{u}^\mathrm{N}}
\def \lamcl{\lambda_\mathrm{g}^\mathrm{L}}

\def \iccl{I_\mathrm{gg}^\mathrm{L}}
\def \iccn{I_\mathrm{gg}^\mathrm{N}}
\def \icc{I_\mathrm{gg}}
\def \icu{I_\mathrm{gu}}
\def \iuc{I_\mathrm{ug}}
\def \iuu{I_\mathrm{uu}}
\def \iuul{I_\mathrm{uu}^\mathrm{L}}
\def \iuun{I_\mathrm{uu}^\mathrm{N}}
\def \icul{I_\mathrm{gu}^\mathrm{L}}
\def \icun{I_\mathrm{gu}^\mathrm{N}}
\def \iucl{I_\mathrm{ug}^\mathrm{L}}
\def \iucn{I_\mathrm{ug}^\mathrm{N}}
\def \lapiccl{\mathcal{L}_{\iccl}}
\def \lapiccn{\mathcal{L}_{\iccn}}
\def \lapicul{\mathcal{L}_{\icul}}
\def \lapicun{\mathcal{L}_{\icun}}
\def \lapicu{\mathcal{L}_{\icu}}
\def \lapiuc{\mathcal{L}_{\iuc}}
\def \lapicc{\mathcal{L}_{\icc}}
\def \lapiucl{\mathcal{L}_{\iucl}}
\def \lapiucn{\mathcal{L}_{\iucn}}
\def \lapiuu{\mathcal{L}_{\iuu}}
\def \lapiuul{\mathcal{L}_{\iuul}}
\def \lapiuun{\mathcal{L}_{\iuun}}
\def \pcb{\mathsf{p}_\mathrm{gb}}
\def \pub{\mathsf{p}_\mathrm{ub}}
\def \pcu{\mathsf{p}_\mathrm{gu}}
\def \pruu{\mathsf{p}_\mathrm{uu}}
\def \prcb{\mathsf{p}_\mathrm{gb}}
\def \fici{\hat{\Phi}_c}
\def \tl{\tau_\mathrm{L}}
\def \tn{\tau_\mathrm{N}}
\def \acbl{\alpha_{\mathrm{gb}}^\mathrm{L}}
\def \acbn{\alpha_{\mathrm{gb}}^\mathrm{N}}
\def \acul{\alpha_{\mathrm{gu}}^\mathrm{L}}
\def \acun{\alpha_{\mathrm{gu}}^\mathrm{N}}
\def \acb{\alpha_{\mathrm{gb}}}
\def \alcb{\alpha_{\mathrm{gb}}^\mathrm{L}}
\def \ancb{\alpha_{\mathrm{gb}}^\mathrm{N}}
\def \aubl{\alpha_{\mathrm{ub}}^\mathrm{L}}
\def \aubn{\alpha_{\mathrm{ub}}^\mathrm{N}}
\def \auu{\alpha_{\mathrm{uu}}}
\def \auul{\alpha_{\mathrm{uu}}^\mathrm{L}}
\def \auun{\alpha_{\mathrm{uu}}^\mathrm{N}}
\def \zl{\zeta_\mathrm{L}}
\def \zn{\zeta_\mathrm{N}}
\def \ec{\epsilon_\mathrm{g}}
\def \eu{\epsilon_\mathrm{u}}
\def \subl{\psi_\mathrm{ub}^\mathrm{L}}
\def \scbl{\psi_\mathrm{gb}^\mathrm{L}}
\def \scbn{\psi_\mathrm{gb}^\mathrm{N}}
\def \scul{\psi_\mathrm{gu}^\mathrm{L}}
\def \suu{\psi_\mathrm{uu}}
\def \suul{\psi_\mathrm{uu}^\mathrm{L}}
\def \suun{\psi_\mathrm{uu}^\mathrm{N}}
\def \zubl{\zeta_\mathrm{ub}^\mathrm{L}}
\def \zubn{\zeta_\mathrm{ub}^\mathrm{N}}
\def \zcul{\zeta_\mathrm{cu}^\mathrm{L}}
\def \zuu{\zeta_\mathrm{uu}}
\def \zuul{\zeta_\mathrm{uu}^\mathrm{L}}
\def \zuun{\zeta_\mathrm{uu}^\mathrm{N}}
\def \zcb{\zeta_\mathrm{gb}}
\def \zcbl{\zeta_\mathrm{gb}^\mathrm{L}}
\def \zcbn{\zeta_\mathrm{gb}^\mathrm{N}}
\def \pur{\rho_\mathrm{u}}
\def \pcr{\rho_\mathrm{g}}
\def \guu{\mathrm{g_{uu}}}
\def \gcu{\mathrm{g_{gu}}}
\def \gub{\mathrm{g_{ub}}}
\def \gubi{\mathrm{g_{ub}^{(i)}}}
\def \gcb{\mathrm{g_{gb}}}
\def \scb{\psi_\mathrm{gb}}
\def \hub{\mathrm{h_{ub}}}
\def \hcb{\mathrm{h_{gb}}}
\def \hcu{\mathrm{h_{gu}}}
\def \dub{d_\mathrm{{ub}}}
\def \dcb{d_\mathrm{{gb}}}
\def \dcu{d_\mathrm{{gu}}}
\def \publ{\Psi_\mathrm{ub}^\mathrm{L}}
\def \pcul{\Psi_\mathrm{gu}^\mathrm{L}}
\def \pcb{\Psi_\mathrm{gb}}
\def \pcbl{\Psi_\mathrm{gb}^\mathrm{L}}
\def \pcbn{\Psi_\mathrm{gb}^\mathrm{N}}
\def \puu{\Psi_\mathrm{uu}}
\def \puul{\Psi_\mathrm{uu}^\mathrm{L}}
\def \puun{\Psi_\mathrm{uu}^\mathrm{N}}
\def \buu{\beta_\mathrm{uu}}
\def \bcb{\beta_\mathrm{gb}}
\def \bubl{\beta_\mathrm{ub}^\mathrm{L}}
\def \Bx{\mathrm{B_x}}
\def \Cx{\mathcal{C}_\mathrm{x}}
\def \Rx{\mathcal{R}_\mathrm{x}}
 
\def \hb{\mathrm{h}_{\mathrm{b}}}
\def \hu{\mathrm{h}_{\mathrm{u}}}
\def \hg{\mathrm{h}_{\mathrm{g}}}
\def \hx{\mathrm{h}_{\mathrm{x}}}
\def \hy{\mathrm{h}_{\mathrm{y}}}
\def \hxy{\mathrm{h}_{\mathrm{xy}}}
\def \x{\mathrm{x}}
\def \y{\mathrm{y}}
\def \zetaxy{\zeta_{\mathrm{xy}}}
\def \tauxy{\tau_{\mathrm{xy}}}
\def \tauoxy{\tau_{0,\mathrm{xy}}}
\def \gxy{\mathrm{g_{xy}}}
\def \psixy{\psi_{\mathrm{xy}}}
\def \psixyL{\psixy^{\mathrm{L}}}
\def \psixyN{\psixy^{\mathrm{N}}}
\def \L{\mathrm{L}}
\def \N{\mathrm{N}}
\def \a{\mathrm{a}}
\def \x{\mathrm{x}}
\def \y{\mathrm{y}}
\def \dxy{d_\mathrm{xy}}
\def \rxy{r_\mathrm{xy}}
\def \alphaxy{\alpha_\mathrm{xy}}
\def \alphaxyL{\alphaxy^\mathrm{L}}
\def \alphaxyN{\alphaxy^\mathrm{N}}
\def \Px{P_{\mathrm{x}}}
\def \Pxmax{P_{\mathrm{x}}^{\textrm{max}}}
\def \rhox{\rho_{\mathrm{x}}}
\def \epsx{\epsilon_{\mathrm{x}}}
\def \mxy{\mathrm{m}_{\mathrm{xy}}}
\def \mxyL{\mxy^{\mathrm{L}}}
\def \mxyN{\mxy^{\mathrm{N}}}
\def \Ixy{I_{\mathrm{xy}}}
\def \Phib{\Phi_{\mathrm{b}}}
\def \Phig{\Phi_{\mathrm{g}}}
\def \Phiu{\Phi_{\mathrm{u}}}
\def \Phihatg{\hat{\Phi}_{\mathrm{g}}}
\def \Phihatgb{\hat{\Phi}_{\mathrm{gb}}}
\def \Phihatgu{\hat{\Phi}_{\mathrm{gu}}}


\def \pcovuav{\mathcal{C}_\mathrm{u}}
\def \pcovc{\mathcal{C}_\mathrm{g}}
\def \rMuu{\mathrm{r_M}}
\def \ru{\mathrm{r_u}}
\def \Ru{R_\mathrm{u}}
\def \u{\mathrm{u}}
\def \c{\mathrm{g}}
\def \q{\mathrm{q}}
\def \r{\mathrm{r}}
\def \m{\mathrm{m}}
\def \prxy{\mathsf{p}_\mathrm{xy}}
\def \prxyxi{\mathsf{p}_\mathrm{xy}^\xi}
\def \pl{\mathcal{P}^\mathrm{L}}
\def \pn{\mathcal{P}^\mathrm{N}}
\def \rc{\mathrm{r_g}}
\def \d{\mathrm{d}}
\def \L{\mathrm{L}}
\def \N{\mathrm{N}}
\def \pu{P_\mathrm{u}}
\def \pumax{\mathrm{P_u^{max}}}
\def \lapixyxi{\mathcal{L}_{I_\mathrm{xy}^\xi}}
\def \px{P_\mathrm{x}}
\def \sy{\mathrm{s_y}}
\def \sixy{\psi_\mathrm{xy}}
\def \zetxy{\zeta_\mathrm{xy}}
\def \prxi{\mathcal{P}^\xi}
\def \s{\mathrm{s}}
\def \axy{\alpha_\mathrm{xy}}
\def \bxy{\beta_\mathrm{xy}}
\def \hxy{\mathrm{h_{xy}}}
\def \Rc{R_\mathrm{g}}
\def \pc{P_\mathrm{g}}
\def \lamc{\lambda_\c}
\def \D{\mathrm{D}}
\def \sigmau{\sigma_\mathrm{u}}
\def \g{\mathrm{g}}


\newtheorem{Theorem}{\bf Theorem}
\newtheorem{Corollary}{\bf Corollary}
\newtheorem{Remark}{\bf Remark}
\newtheorem{Lemma}{\bf Lemma}
\newtheorem{Proposition}{\bf Proposition}
\newtheorem{Assumption}{\bf Assumption}
\newtheorem{Approximation}{\bf Approximation}
\newtheorem{Definition}{\bf Definition}

\title{Spectrum Sharing Strategies for\\UAV-to-UAV Cellular Communications}
\author{\IEEEauthorblockN{M.~Mahdi~Azari$^{\star}$, Giovanni Geraci$^{\diamond}$, Adrian Garcia-Rodriguez$^{\dag}$, and Sofie~Pollin$^{\ast}$}\\ \vspace{-0.3cm}
\normalsize\IEEEauthorblockA{\emph{$^{\star}$CTTC, Spain \enspace $^{\diamond}$Universitat Pompeu Fabra, Spain \enspace $^{\dag}$Nokia Bell Labs, Ireland
\enspace $^{\ast}$KU Leuven, Belgium}}
\thanks{The work of M.~M.~Azari was partly supported by the Catalan government under grant 2017 SGR1479. The work of G.~Geraci was partly supported by MINECO under Project RTI2018-101040-A-I00 and by the Postdoctoral Junior Leader Fellowship Programme from ``la Caixa" Banking Foundation.}
}
\maketitle
\thispagestyle{empty}
\begin{abstract}
In this article, we consider a cellular network deployment where UAV-to-UAV (U2U) transmit-receive pairs coexist with the uplink (UL) of cellular ground users (GUEs). Our analysis focuses on comparing two spectrum sharing mechanisms: i) overlay,  where the available time-frequency resources are split into orthogonal portions for U2U and GUE communications, and ii) underlay, where the same resources may be accessed by both link types, resulting in mutual interference.
We evaluate the coverage probability and rate of all links and their interplay to identify the best spectrum sharing mechanism. Among other things, our results demonstrate that, in scenarios with a large number of UAV pairs, adopting overlay spectrum sharing seems the most suitable approach for maintaining a minimum guaranteed rate for UAVs and a high GUE UL performance. We also find that increasing the density of U2U links degrades their rates in the overlay---where UAVs only receive interference from other UAVs---, but not significantly so in the underlay---where the effect of GUE-generated interference is dominant.
\end{abstract}
\IEEEpeerreviewmaketitle
\section{Introduction}
\label{sec:Intro}

From a social perspective, enabling cellular-connected unmanned aerial vehicles (UAVs) is expected to bring important benefits in terms of logistics automation, facilitating search-and-rescue missions, and even coverage and capacity enhancements through mobile small cells \cite{geraci2019preparing,azari2017ultra,zeng2019accessing,MozSaaBen2018,vinogradov2019tutorial}. From a business point of view, offering cellular coverage to aerial users could translate into new revenue opportunities for mobile network operators \cite{Ericsson:18,fotouhi2018survey}. 

While present-day networks should be able to support cellular-connected UAVs up to a certain extent \cite{azari2019cellular,azari2017coexistence,AzaRosPol2017,NguAmoWig2018,MeiWuZhang2018,LopDinLi2018GC,amer2019toward}, 5G-and-beyond hardware and software upgrades may be required by both mobile operators and UAV manufacturers to accommodate for many connected UAVs \cite{GarGerLop2018,GerGarGal2018,DanGarGerICC2019}. For some use-cases---including collision avoidance and autonomous UAV swarm operations---it could be desirable to bypass the ground network infrastructure and enable direct UAV-to-UAV (U2U) communications \cite{AzaGerGar19,ZhaZhaDi19}.

In this paper, we consider a cellular network deployment where UAV transmit-receive pairs share the same spectrum with the uplink (UL) of cellular ground users (GUEs). To address the interference issues originally identified in \cite{AzaGerGar19}---where GUE UL and U2U communications fully reused the entire set of time-frequency physical resource blocks (PRBs)---, we examine two strategies for spectrum sharing, namely \emph{overlay} and \emph{underlay}. In the overlay, the available PRBs are split into two orthogonal portions, respectively reserved for each link type. In the underlay, UAVs are allowed to access a fraction of the PRBs available for the GUE UL, resulting in mutual interference.


For the first time, we compare the performance of U2U links and GUE UL with both spectrum sharing mechanisms, and evaluate the impact that the number of PRBs accessed by each link type has on the coexistence of aerial and ground communications. To do so, we adopt a realistic channel model and system setup, we propose and validate new tight approximations, and we obtain compact analytical expressions through stochastic geometry tools.


The key conclusions of our analysis are as follows:
\begin{itemize}[leftmargin=*]
\item In an urban scenario, implementing an overlay spectrum sharing approach is the best option in order to maintain a high GUE UL performance while guaranteeing a minimum rate of 100 kbps to the majority of U2U pairs.
\item In the overlay, since UAVs are only interfered by other UAVs, higher UAV densities result in lower U2U rates. On the other hand, in the underlay, where GUE-to-UAV interference is dominant, the rate degradation at UAVs caused by increasing their density is limited.
\item In the underlay, increasing the number of PRBs utilized by UAV pairs causes a sharp performance degradation for GUEs, unless both the UAV density and the UAV transmission power are limited.
\end{itemize}

\section{System Model}
\label{sec:System_Model}

\begin{figure*}[!t]
\centering
\includegraphics[width=1.75\columnwidth]{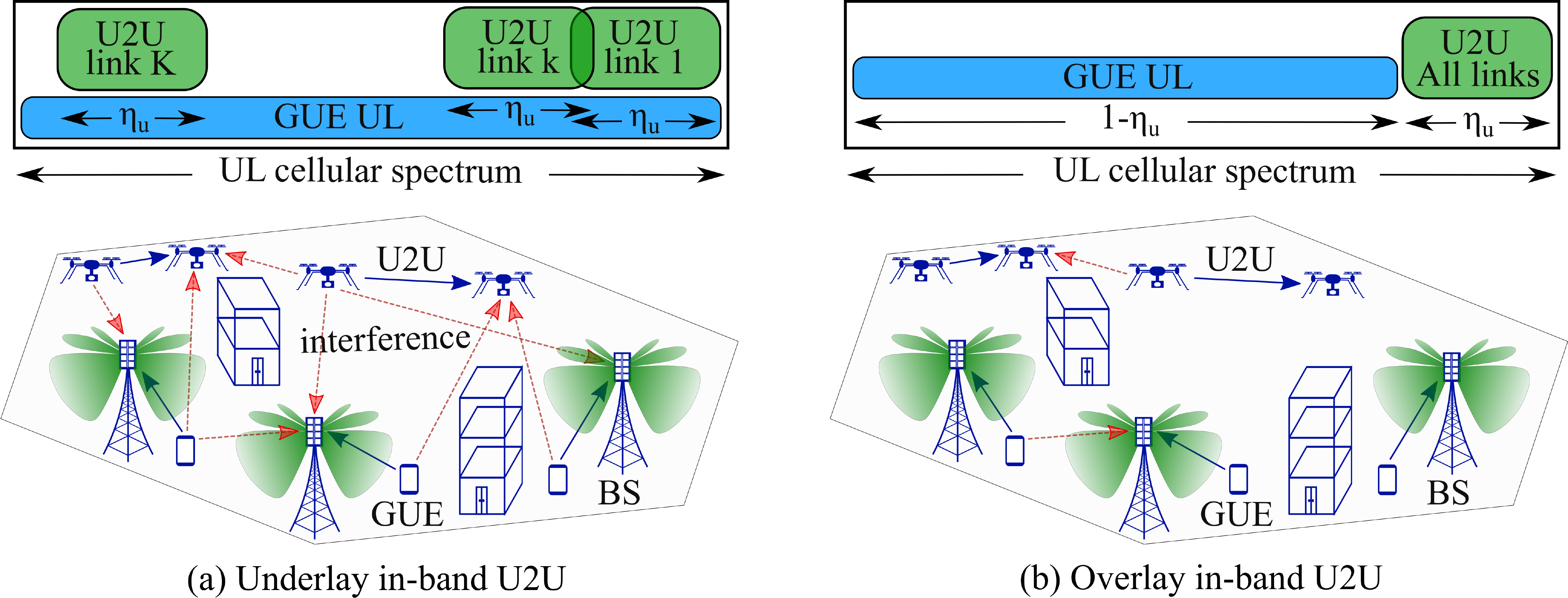}
	\caption{U2U communications sharing spectrum with the cellular UL. Blue solid (resp. red dashed) arrows indicate communication (resp. interfering) links. In (a)---underlay in-band U2U---GUEs occupy the whole spectrum while UAVs occupy a fraction $\eta_\u$, incurring mutual GUE-U2U interference. In (b)---overlay in-band U2U---the spectrum is split into orthogonal portions, with a fraction $\eta_\u$ reserved to UAVs and no mutual interference.}
	\label{U2U_SystemModel}
\end{figure*}


\subsection{Network Topology}

We consider a cellular network as in Fig.~\ref{U2U_SystemModel}, where base stations (BSs) are deployed at a height $\hb$ and uniformly distributed as a Poisson point process (PPP) $\Phib \in \mathbb{R}^2$ with density $\lamb$. We assume that GUE UL transmissions and U2U direct communications reuse the same spectrum.

\subsubsection*{GUE UL communications}
GUEs transmit towards their closest BS\footnote{A GUE may connect to a BS $b$ other than the closest one $a$ if its link is in line-of-sight (LoS) condition with $b$ and not with $a$. This is however unlikely, since the probability of LoS decreases with the distance \cite{3GPP36777}.} and form an independent PPP $\Phig \in \mathbb{R}^2$ with density $\lamg=\lamb$ on each PRB. Accordingly, the 2-D distance between a GUE and its serving BS follows a Rayleigh distribution with parameter $\sigma_\c = 1/\sqrt{2\pi \lamc}$ \cite{ChuCotDhi2017}. For the typical BS, the interfering GUEs form a non-homogeneous PPP with density $\lamci(r)\!=\!\lamb(1\!-\!e^{-\lamb \pi r^2})$, function of their 2-D distance $r$ \cite{ChuCotDhi2017,YanGerQue:16}.

\subsubsection*{Direct UAV-to-UAV communications}
We assume that U2U transmitters form a PPP $\Phiu$ with density $\lamu$, and that each U2U receiver is randomly and independently placed around its corresponding transmitter, both at height $\hu$ and with distance $\Ru$ following a truncated Rayleigh distribution with probability density function (PDF) \cite{AzaGerGar19}
\begin{align}
f_{\Ru}(\ru) = \frac{\ru e^{-\r_\u^2/(2\sigma_\mathrm{u}^2)}}{\sigma_\mathrm{u}^2\left(1-\e^{-r_\mathrm{M}^2/(2\sigma_\mathrm{u}^2)}\right)}\cdot\mathds{1}(\ru<r_\mathrm{M}),
\end{align}
where $\r_\mathrm{M}$ is the maximum U2U link distance, $\mathds{1}(\cdot)$ is the indicator function, and $\sigma_\mathrm{u}$ is the Rayleigh scale parameter, related to the mean distance $\bar{R}_\mathrm{u}$ through $\sigma_\mathrm{u} = \sqrt{2 / \pi}\bar{R}_\mathrm{u}$.

\subsection{Propagation Channel and Antenna Pattern Models}

We assume that any radio link between nodes $\x$ and $\y$ is affected by large-scale fading $\zetaxy$, comprising path loss $\tauxy$ and antenna gain $\gxy$, and small-scale fading $\psixy$.

\subsubsection*{Probability of LoS}
We consider that links experience LoS and non-LoS (NLoS) propagation conditions with probabilities $\prxy^\L$ and $\prxy^\N$, respectively. In what follows, the superscripts $\mathrm{L}$ and $\mathrm{N}$ will denote system parameters under LoS and NLoS conditions, respectively. We employ the well-known ITU model to characterize the probability of LoS between any pair of nodes x and y \cite{ITU2012}
\begin{equation}
\prxy^\L(r) \!=\!\!\!\!\!\!\!\! \prod_{j=0}^{\lfloor{\frac{r\sqrt{\a_1\a_2}}{1000}-1\rfloor}} \!\!\Bigg[1\!-\!\mathrm{exp}{\Bigg(\!\!\!-\!\frac{\left[\!\hx\!-\!\frac{(j+0.5)(\hx-\hy)}{\mathrm{k}+1} \!\right]^2}{2 \a_3^2}\Bigg)} \!\Bigg]\!,\!\!
\label{PrLoS}
\end{equation}
where $\a_1 = 0.3$, $\a_2 = 500$, and $\a_3\ = 20$ model an urban scenario, and approximate $\prxy^\L$ with a step function, constant in $[\r_i,\r_{i+1}]$, $i = 1,2,\ldots$, $0=\r_1 < \r_2 <\ldots$

\subsubsection*{Path loss}
The path loss between nodes $\mathrm{x}$ and $\mathrm{y}$ is
\begin{equation}
\tauxy = \hat{\tau}_\mathrm{xy} \, \dxy^{\,\alphaxy},
\end{equation}
where $\hat{\tau}_\mathrm{xy}$ denotes the reference path loss, $\alphaxy$ is the path loss exponent, and $\dxy \!=\! \sqrt{\rxy^2 \!+\! \mathrm{h^2_{xy}}}$, $\rxy$, and $\hxy \!=\! \hx \!-\! \hy$ represent the 3-D distance, 2-D distance, and height difference between $\mathrm{x}$ and $\mathrm{y}$, respectively. In the following, we employ the subscripts $\{\u,\c,\mathrm{b}\}$ to denote UAV, GUE, and BS nodes, respectively.

\subsubsection*{Antenna models}
Let $\gxy$ denote the antenna gain between nodes $\x$ and $\y$. We consider that GUEs and UAVs have a single omnidirectional antenna with 0~dBi gain. Instead, the overall antenna radiation pattern of each BS at zenith angle $\theta$ is given by $g_b(\theta) = g_A(\theta) \cdot g_E(\theta)$, where
\begin{equation}
g_A(\theta) = \frac{\sin^2\Big(N\pi (\cos\theta - \cos\tt)/2\Big)}{N\sin^2\Big(\pi (\cos\theta - \cos\tt)/2\Big)}
\end{equation}
is the array factor of an $\mathrm{N}$-element uniform vertical linear array, and 
\begin{equation} \label{ElementGain}
g_E(\theta) = g_E^{\max} \sin^2\theta
\end{equation}
represents the BS individual antenna element gain. 

\subsubsection*{Small-scale fading}
We adopt the general Nakagami-m model, where the cumulative distribution function (CDF) of the small-scale fading power between nodes $\x$ and $\y$ can be expressed as
\begin{equation} \label{FadingCDF}
F_{\sixy}(\omega) \triangleq \mathbb{P}[\sixy < \omega] \!=\! 1\!-\!\sum_{i=0}^{\mxy-1} \!\frac{(\mxy \omega)^i}{i!} e^{-\mxy \omega},
\end{equation}
where $\mxy \in \mathbb{Z}^{+}$ is the fading parameter.

\subsection{Spectrum Sharing and Power Control Policies}

Let the available spectrum be divided into $\mathrm{n}$ PRBs. We consider two spectrum sharing strategies---underlay and overlay---, illustrated in Fig.~\ref{U2U_SystemModel} and described as follows.

\subsubsection*{Underlay in-band U2U}

Each PRB may be used by both link types \cite{LinAndGho2014}. In particular, we assume that:
\begin{itemize}[leftmargin=*]
\item Each active GUE occupies all $\mathrm{n}$ PRBs. This is consistent with a cellular operator's goal of preserving the performance of its legacy ground users \cite{GerGarGal2018}.
\item Each U2U transmitter occupies a fraction $\eta_\u$ of all PRBs, also employing frequency hopping to randomize its interference to other links. Specifically, each U2U transmitter may randomly and independently access $\eta_\u \cdot \mathrm{n}$ PRBs, where the factor $\eta_\u \in [0,1]$ measures the aggressiveness of the U2U spectrum access, and is denoted the spectrum access factor in the underlay. As a result, the density of interfering UAVs is $\hat{\lambda}_\u = \eta_\u \cdot \lamu$.
\end{itemize}

\subsubsection*{Overlay in-band U2U}

The available UL spectrum is split into two orthogonal portions. A fraction $\eta_\u$ is reserved for U2U communications, and UAVs access all $\eta_\u \cdot \mathrm{n}$ allocated PRBs without frequency hopping. Similarly, the remaining fraction $\eta_\g = 1 - \eta_\u$ is reserved to the GUEs UL, and active GUEs access all $\eta_\g \cdot \mathrm{n}$ PRBs allocated. We denote $\eta_\u$ the spectrum partition factor in the overlay. This approach results in each GUE UL link being interfered only by other GUEs, and in each U2U link being interfered only by other UAVs.

\subsubsection*{Fractional power control}
As per the cellular systems currently deployed, the power transmitted per PRB by a given node $\x$ is adjusted depending on the receiver $\y$ and can be computed as \cite{BarGalGar2018GC}
\begin{equation}
\Px = \min\left\{ \Pxmax, \rhox \cdot \zetaxy^{\epsx} \right\},
\label{eqn:power_control}
\end{equation}
where $\Pxmax$ is given by the maximum transmit power constraint divided by the number of PRBs utilized by node $\x$ for transmission, i.e., $P_{\mathrm{max}} / \mathrm{n_x}$, whereas $\rhox$ is a parameter adjusted by the network and $\epsx \in [0,1]$ is the fractional power control factor \cite{3GPP36777}.

\section{Analytical Results}
\label{sec:analysis_underlay}



In what follows, we will analyze the coverage probability, denoted by $\mathcal{C}_\x$ for node $\x$.\footnote{Due to space constraints, proofs are omitted and available in \cite{azari2019uav}.} This is defined as the complementary CDF (CCDF) of the signal-to-interference-plus-noise-ratio (SINR), i.e., the probability of the SINR at node $\x$, SINR$_\x$, being beyond a certain threshold $\t$:
\begin{equation}
	\mathcal{C}_\x \triangleq \mathbb{P}\{\sinr_\x>\t\}.
\end{equation}
From $\mathcal{C}_\x$, the coverage rate probability can be obtained as the CCDF of the achievable rate $\Rx$ at node $\x$ \cite{bai2015coverage}:
\begin{equation}
\mathbb{P}[\Rx>\t]=\Cx(2^{\t/\Bx}-1),
\end{equation}
with $\Bx$ denoting the bandwidth accessed by node $\x$.


\subsection{Preliminaries}

In order to obtain more compact analytical expressions, we employ the following approximations whose accuracy will be validated in Section~\ref{sec:numerical}.

\begin{Approximation} \label{FadingCDFapp}
	We approximate the CDF of the Nakagami-m small-scale fading power $\sixy$ in \eqref{FadingCDF} as
	\begin{equation} \label{eq:CDFapp}
	F_{\sixy}(\omega) \approx \left( 1- e^{-\mathrm{b_{xy}}\,\omega}\right)^{\mxy},
	\end{equation}
	where $\mathrm{b_{xy}}$ is a function of $\mxy$ provided in \cite[Table II]{azari2019uav}.
	
	
\end{Approximation}

Approximation~\ref{FadingCDFapp} is inspired by \cite{bai2015coverage} and allows to derive closed-form expressions for the Laplacian of the interference, and in turn for the coverage probability. The value of $\mathrm{b_{xy}}$ is obtained through curve fitting. 

\begin{Approximation} \label{NLoSapp}
We neglect the NLoS UAV-to-UAV, GUE-to-UAV, and UAV-to-BS interfering links.
\end{Approximation}

Approximation~\ref{NLoSapp} holds due to a high probability of having LoS links dominating the interference \cite{azari2019cellular,GerGarGal2018}.

\begin{Approximation} \label{MeanPowerapp}
We approximate the UAVs transmit power, which is a random variable, with its mean value.
\end{Approximation}	

Approximation~\ref{MeanPowerapp} is motivated by the fact that U2U links tend to undergo LoS conditions, and thus a lower path loss exponent. This implies a lower variation of the UAV transmit power with respect to its distance from the receiver. This approximation removes one integral in the computation of the coverage probability. The mean UAV transmit power is obtained as follows.
	
\begin{Proposition} \label{proposition:meanUAVtxPower}
	The mean UAV transmit power is given by
	\begin{align} \label{eqn:meanPower}
	\mathbb{E}[\pu] =  \!\!\! \sum_{\nu \in \{\mathrm{L},\mathrm{N}\}}  &\Big[\sum_{i=1}^{j} [C_i^\nu-C_{i+1}^\nu] \,\, \gamma\Big(1+\frac{\auu^\nu \eu}{2},y_{i+1}\Big)\nonumber\\
	&+\sum_{i=j+1}^{k+1} [B_i^\nu-B_{i-1}^\nu] \, e^{-y_i} \Big]
	\end{align}
	where
	\begin{align} 
	C_i^\nu &\!=\! \frac{(2\sigma_\u^2)^{{\auu^\nu\eu/2}}\pur\left(\hat{\tau}_{\mathrm{uu}}^\nu/\guu\right)^{\eu}}{1-\e^{-r_\mathrm{M}^2/(2\sigma_\mathrm{u}^2)}} \cdot \pruu^\nu(r_i), \text{ for } i>0,
	\end{align}
	$B_j^\nu = 0$, $B_{k+1}^\nu = 0$, and
	\begin{align}
	B_i^\nu &= \frac{\pumax \, \pruu^\nu(r_i) }{1-\e^{-r_\mathrm{M}^2/(2\sigma_\mathrm{u}^2)}};~i>j,
	\end{align}
	and where $j = \lfloor{\frac{r_m^\nu\sqrt{\a_1\a_2}}{1000}\rfloor}+1$, $k = \lfloor{\frac{r_M\sqrt{\a_1\a_2}}{1000}\rfloor}+2$, $y_i = r_i^2 / 2\sigma_u^2$, $r_k = r_M$, and $r_{j+1} = r_m^\nu$. The latter is the distance at which the UAV reaches its maximum allowed transmit power, which depends on the link condition (LoS vs. NLoS) and can be obtained from \eqref{eqn:power_control} as follows
	\begin{equation}
	r_m^\nu = \left(\frac{\guu}{\hat{\tau}_{\mathrm{uu}}^\nu}\right)^{1/\auu^\nu} \cdot \left(\frac{\pumax}{\pur}\right)^{1/(\auu^\nu\eu)}.
	\end{equation}
\end{Proposition}

\begin{proof}
Available in \cite[Appendix E]{azari2019uav}.
\end{proof}

\subsection{U2U Coverage Probability: Underlay and Overlay} \label{sec:U2U_performance}

We now analyze the performance of a U2U communication link in the underlay and in the overlay.

\begin{Theorem} \label{proposition:U2Ucoverage}
	Under Approximations~1-3, the underlay U2U coverage probability is given by
	\begin{align} \label{eq:U2ULinkCoverage}
	\pcovuav = \int_0^\rMuu f_{\Ru}^\L(\ru) \mathcal{C}_{\mathrm{u}|\Ru}^\L(\ru) \mathrm{d}\ru,
	\end{align}
	where
	\begin{equation}
	\mathcal{C}_{\mathrm{u}|\Ru}^\L(\ru) = \sum_{i=1}^{\mathrm{m_{uu}^\L}} \binom{\mathrm{m_{uu}^\L}}{i}(-1)^{i+1} e^{-z_{\u,i}^\L \mathrm{N_0}} \cdot \lapiu^\L(z_{\u,i}^\L),
	\end{equation}
	\begin{align}
	\lapiu^\L(z_{\u,i}^\L) = \underbrace{e^{ -2 \pi \hat{\lambda}_\u \mathcal{I}_\mathrm{uu}^\L}}_{\text{due to LoS UAVs}} \cdot \underbrace{e^{ -2 \pi \lamb \mathcal{I}_\mathrm{gu}^\L }}_{\text{due to LoS GUEs}},
	\end{align}
	and
	\begin{align}
	\mathcal{I}_\mathrm{uu}^\L &= \sum_{j = 1}^\infty \left[\pruu^{\L}(\r_{j-1})-\pruu^{\L}(\r_{j})\right] \underbrace{\Psi_\mathrm{uu}^\L\left(\mathrm{s},\r_{j}\right)}_\text{at $\pu = \bar{P}_\mathrm{u}$},
	\end{align}
	with $\mathcal{I}_\mathrm{gu}^\L$ and $\Psi_\mathrm{uu}^\L$ obtained from \cite[Theorem~1]{AzaGerGar19} and
	\begin{equation}
	\mathrm{s} = z_{\u,i}^\L \frac{\mathrm{g_{uu}}}{\hat{\tau}_{\mathrm{uu}}^\L};
	~~z_{\u,i}^\L = \frac{i b_{\u\u}^\L \t}{\pu^\L \zuu^\L(\ru)^{-1}}.
	\end{equation} 
\end{Theorem}

\begin{proof}
Available in \cite[Appendix D]{azari2019uav}.
\end{proof}

\begin{Corollary}
Under Approximations~1-3, the overlay U2U coverage probability can be obtained from Theorem~\ref{proposition:U2Ucoverage} by substituting $\lamb=0$, $\lamuhat=\lamu$, and
	\begin{equation}
	\lapiu^\L(z_{\u,i}^\L) = e^{-2 \pi \lamu \mathcal{I}_\mathrm{uu}^\mathrm{\L}(z_{\u,i}^\L)  },
	\end{equation}
since the aggregate interference only includes the UAV-generated one.
\end{Corollary}

\begin{Remark} \label{UAVpowerLoS}
When the UAV altitude is sufficiently high, UAV-UAV links become almost always LoS as per \eqref{PrLoS}, and we can employ $\pruu^\L(\r_i) = 1$ for any $i$ and hence $\mathcal{I}_\mathrm{uu}^\L = - \Psi_\mathrm{uu}^\L\left(\mathrm{s},0 \right)$. Moreover, the mean UAV transmit power in (\ref{eqn:meanPower}) can be simplified as follows
	\begin{equation}
	\begin{aligned}
	\mathbb{E}[\pu] &= C_1^\L \,\, \gamma\left(1+\frac{\auu^\L \eu}{2},\frac{1}{2}\left(\frac{r_m^\L}{\sigma_u}\right)^2\right) \\
	&\!\!\!\!\!\!\!\!\!\!\!\!\!\!\!\!+ \frac{\pumax }{1-\e^{-r_\mathrm{M}^2/(2\sigma_\mathrm{u}^2)}}\left(e^{-\frac{1}{2}\left({r_m^\L/\sigma_u}\right)^2} - e^{-\frac{1}{2}\left({r_M/\sigma_u}\right)^2}\right).
	\end{aligned}
	\end{equation}
\end{Remark}


\subsection{GUE Coverage Probability: Underlay and Overlay} \label{sec:GUE_performance}

The coverage probability of the GUE UL in the underlay and in the overlay is obtained as follows.

\begin{Theorem} \label{proposition:GUEcov}
Under Approximations~1-3, the underlay GUE UL coverage probability is given by
	\begin{align} 
	\pcovc &= \!\!\!\! \sum_{\nu \in \{\mathrm{L},\mathrm{N}\}}\int_0^\infty f_{R_\c}^\nu(\rc)\, \mathcal{C}_{\c|\Rc}^\nu(\rc)  \,\,\mathrm{d}\rc,
	\end{align}
	where
	\begin{align} 
	\mathcal{C}_{\c|\Rc}^\nu(\rc) &= \sum_{i=1}^{\mathrm{m_{gb}^\nu}} \binom{\mathrm{m_{gb}^\nu}}{i}(-1)^{i+1} e^{-z_{\g,i}^\nu \mathrm{N_0}} \cdot \lapic^\nu(z_{\g,i}^\nu),
	\end{align}
	and  
	\begin{align} \label{eqn:Lig}
	\lapic^\nu(z_{\g,i}^\nu) &= \underbrace{e^{ -2 \pi \hat{\lambda}_\u \mathcal{I}_\mathrm{ug}^\L }}_{\text{due to LoS UAVs}} \cdot \underbrace{e^{ -2 \pi \lamb \sum_{\xi \in \{\mathrm{L},\mathrm{N}\}}\mathcal{I}_\mathrm{gg}^\mathrm{\xi} }}_{\text{due to GUEs}}.
	\end{align}
	In (\ref{eqn:Lig}), $\mathcal{I}_\mathrm{ug}^\L$ is given by
	\begin{align}
	\mathcal{I}_\mathrm{ug}^\L &= \sum_{j = 1}^\infty \pub^{\L}(\r_j) \Big(\underbrace{\Psi_\mathrm{ub}^\L\left(\mathrm{s},\r_{j+1}\right) - \Psi_\mathrm{ub}^\L\left(\mathrm{s},\r_{j}\right)}_\text{at $\pu = \bar{P}_\mathrm{u}$}\Big),
	\end{align}
	whereas $\mathcal{I}_\mathrm{gg}^\mathrm{\xi}$ and $\Psi_\mathrm{ub}^\L$ are provided in \cite[Theorem~2]{AzaGerGar19} where we replace $\mathrm{s}$ with
	\begin{equation}
	z_{\g,i}^\nu = \frac{i b_{\c\mathrm{b}}^\nu \t}{\pc^\nu \zcb^\nu(\ru)^{-1}}.
	\end{equation}
\end{Theorem}

\begin{proof}
Similar to proof of Theorem~\ref{proposition:U2Ucoverage} and omitted.
\end{proof}

\begin{Corollary}
Under Approximations~1-3, the overlay GUE UL coverage probability can be obtained from Theorem~\ref{proposition:GUEcov} by replacing $\hat{\lambda}_\u = 0$, since the aggregate interference only includes the GUE-generated one.
\end{Corollary}

\section{Numerical Results and Discussion}
\label{sec:numerical}

\begin{table}
	\centering
	\caption{System Parameters}
	\label{table:parameters}
	\def\arraystretch{1.2}
	\begin{tabulary}{\columnwidth}{ |p{2.2cm} | p{5.2cm} | }
		\hline
		BS distribution		& PPP with $\lamb = 5$~/~Km$^2$, $\hb=25$~m\\ \hline
		GUE distribution 				& One active GUE per cell, $\hc=1.5$~m \\ \hline
		UAV distribution 				& $\lamu\!=\!1$\,/\,Km$^2$, $\bar{R}_\mathrm{u}\!=\!100\,$m, $\hu\!=\!100\,$m \\ \hline
		\multirow{2}{*}{Ref. path loss [dB]}		&  $\hat{\tau}_\mathrm{cb}^\L = 28+20\log_{10}(f_c)$ \enspace ($f_c$ in GHz) \\ \cline{2-2}
		& $\hat{\tau}_\mathrm{gb}^\N = 13.54+20\log_{10}(f_c)$ \\ \cline{2-2}
		& $\hat{\tau}_\mathrm{ub}^\L = 28+20\log_{10}(f_c)$ \\ \cline{2-2}
		& $\hat{\tau}_\mathrm{ub}^\N = -17.5+20\log{10}(40\pi f_c/3)$ \\ \cline{2-2}
		& $\hat{\tau}_\mathrm{gu}^\L = 30.9+20\log_{10}(f_c)$  \\ \cline{2-2}
		&  $\hat{\tau}_\mathrm{gu}^\N = 32.4+20\log_{10}(f_c)$  \\ \cline{2-2}
		& $\hat{\tau}_\mathrm{uu}^\L = 28+20\log_{10}(f_c)$\\ \cline{2-2}
		& $\hat{\tau}_\mathrm{uu}^\N = -17.5+20\log{10}(40\pi f_c/3)$  \\ \hline
		\multirow{2}{*}{Path loss exponent}		&  $\alcb = 2.2,~~~\ancb = 3.9$ \\ \cline{2-2}
		& $\aubl = 2.2,~~~\aubn = 4.6-0.7\log_{10}(\hu)$ \\ \cline{2-2}
		& $\acul = 2.225-0.05\log_{10}(\hu)$ \\ & $\acun = 4.32-0.76\log_{10}(\hu)$ \\ \cline{2-2}
		& $\auul = 2.2,~~~\auun = 4.6-0.7\log_{10}(\hu)$\\ \hline
		Small-scale fading  & Nakagami-m with $\mxy^\xi = 1$ for NLoS links, $\mxy^\xi = 3$ for LoS GUE links, and $\mxy^\xi = 5$ for LoS UAV links \\ \hline
		Prob. of LoS & ITU model as per \eqref{PrLoS} \\ \hline
		Thermal noise 				& -174 dBm/Hz with 7~dB noise figure \\ \hline 
		\multirow{2}{*}{Spectrum}		& Carrier frequency: 2~GHz \\ \cline{2-2} 
		& Bandwidth: 10 MHz with 50 PRBs \\ \hline 
		BS array configuration		& $8\times 1$ vertical, 1 RF chain, downtilt: $102^{\circ}$, element gain as in \eqref{ElementGain}, spacing: $0.5~\lambda$\\ \hline
		Power control		& Fractional, based on GUE-to-BS (resp. U2U) large-scale fading for GUEs (resp. UAVs), with $\ec = \eu = 0.6$, $\pcr = \pur = -58$~dBm, and $P^{\textrm{max}}_\c = P^{\textrm{max}}_\u = 24$~dBm \cite{BarGalGar2018GC}\\ \hline
		GUE/UAV antenna 		& Omnidirectional with 0~dBi gain \\ \hline 
		\end{tabulary}
\end{table}

In this section, we first validate our analysis and then characterize the performance of U2U and UL GUE cellular communications under overlay and underlay spectrum sharing strategies. For both spectrum sharing mechanisms, we concentrate on evaluating the impact that the number of PRBs allocated to UAVs has on both aerial and ground communications in an urban scenario.\footnote{Unlike \cite{AzaGerGar19}, the results in this section are obtained under general Nakagami-m fading and through compact, tight approximations.} Unless otherwise specified, the system parameters are included in Table~\ref{table:parameters} and follow the 3GPP specifications in \cite{3GPP36777}.


\begin{figure}[t!]
	\centering
	\includegraphics[width=\figwidth]{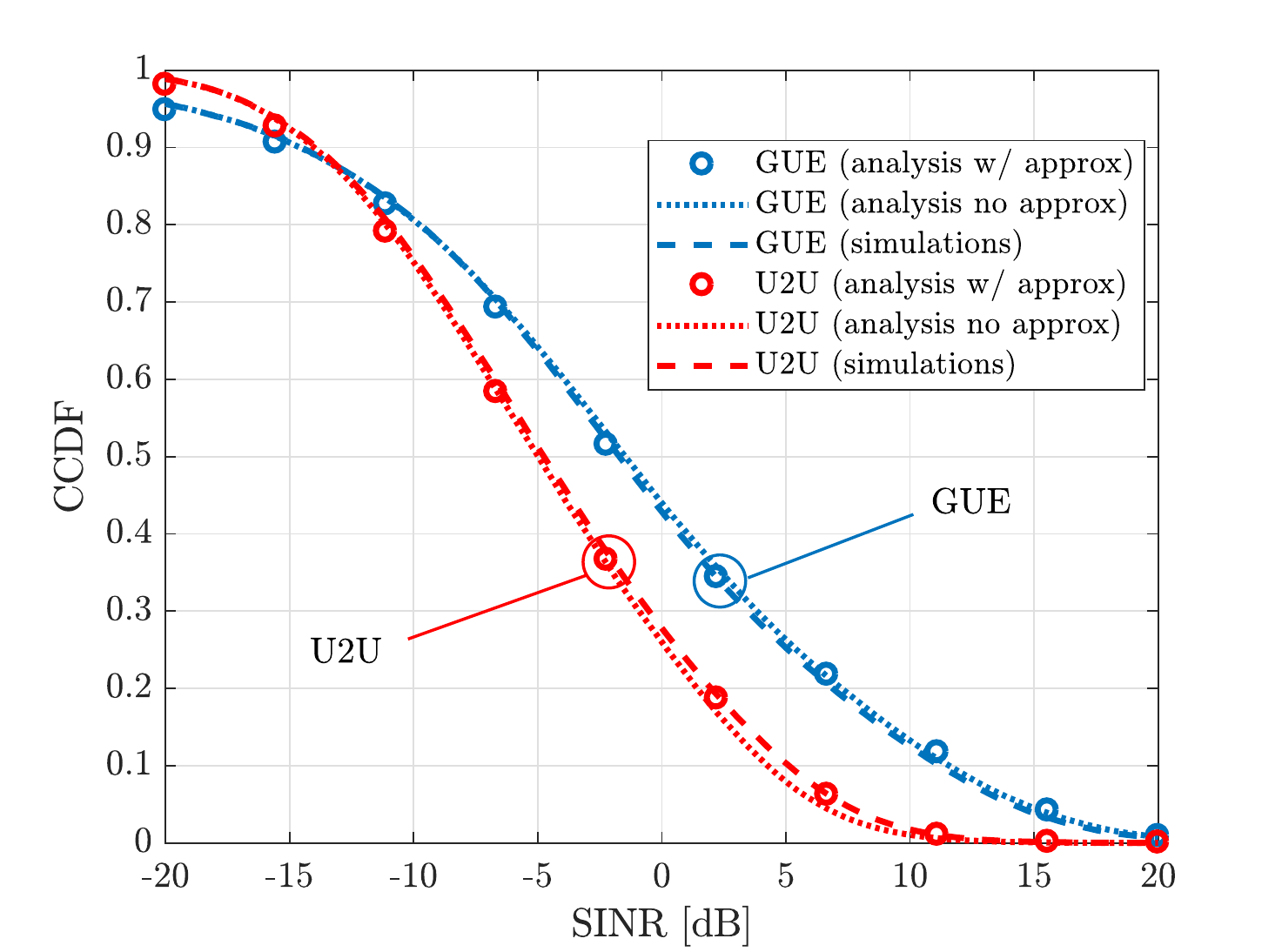}
	\caption{Underlay coverage probability obtained via approximated analysis (solid), exact analysis (dotted), and simulations (dashed).}
	\label{Validation}
\end{figure}

For the purpose of validation, Fig.~\ref{Validation} shows the coverage probability for GUE UL and U2U links in the underlay, with $\eta_\u=1$, obtained in three different ways: i) with our approximated analysis in Section~\ref{sec:analysis_underlay}, ii) through an exact analysis---omitted due to space constraints and available in \cite{azari2019uav}---, and iii) via simulations. The three sets of curves exhibit a close match, thus validating our analysis for the underlay and, as a special case, for the overlay too.

\begin{figure}[t!]
	\centering
	\includegraphics[width=\figwidth]{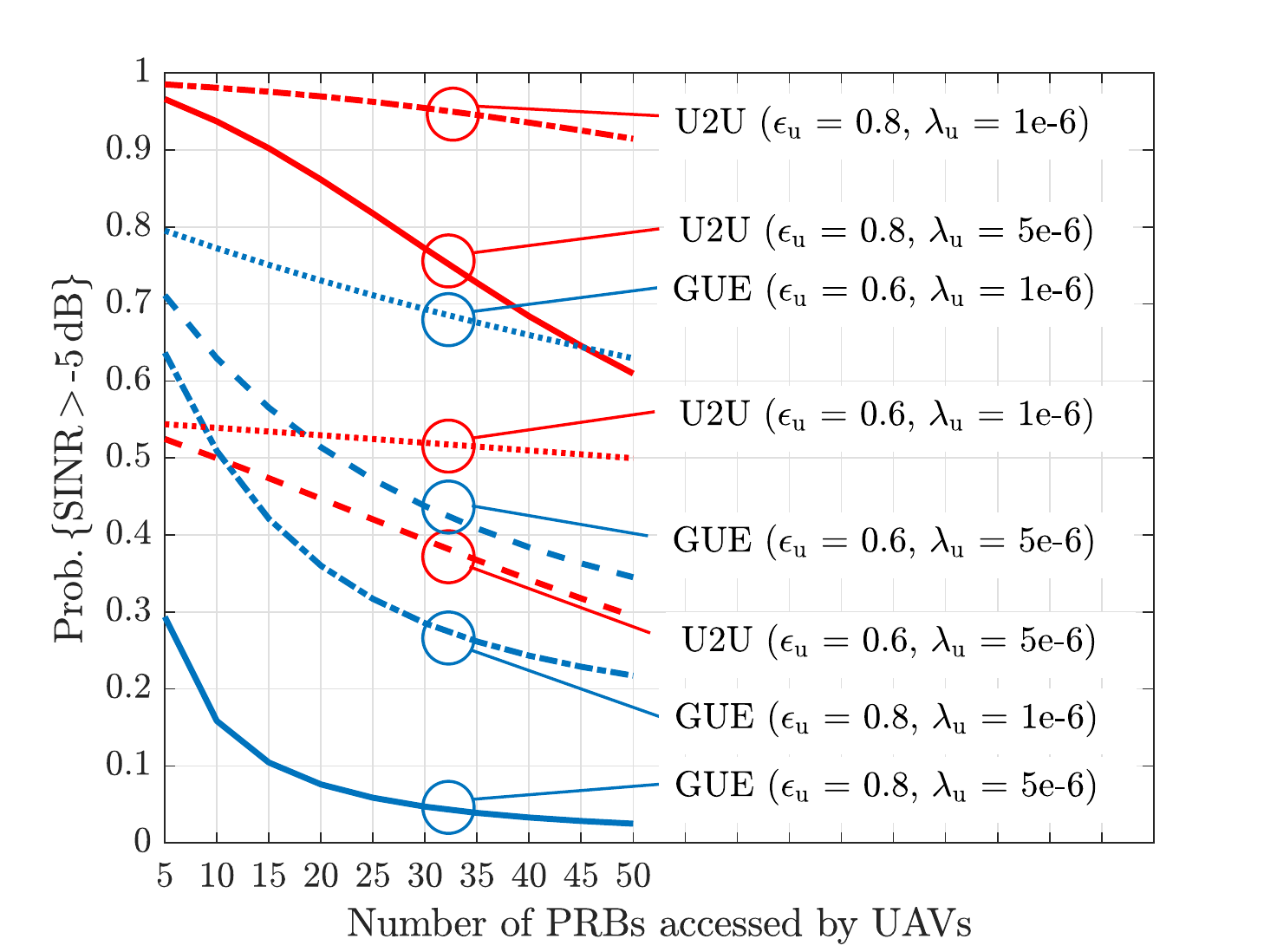}
	\caption{Coverage probability for U2U and GUE UL links in the underlay for increasing $\eta_\u$---i.e., increasing number of PRBs accessed by UAVs---, and various values of $\epsilon_\u$ and $\lambda_\u$.}
	\label{Pcov_nPRBu_underlay_combined}
\end{figure}

Fig.~\ref{Pcov_nPRBu_underlay_combined} shows the probability of experiencing SINRs per PRB larger than -5~dB for the GUE UL and U2U links in the underlay strategy as a function of the number of PRBs accessed by the UAVs. We consider various values for the UAV fractional power control factor $\eu$ and for the UAV density $\lambda_{\mathrm{u}}$. Notably, the results of Fig.~\ref{Pcov_nPRBu_underlay_combined} demonstrate how increasing the number of PRBs allocated to UAV pairs causes a sharp performance degradation for GUEs, except for the case where both the UAV density and the UAV transmit powers are constrained ($\lambda_{\mathrm{u}}=1$e-6, $\epsilon_{\mathrm{u}}=0.6$). As expected, also increasing the UAV density or transmit power generates more interference to the GUE UL, reducing the SINR. As for the U2U link performance vs. the number of PRBs accessed, this remains almost constant for $\lambda_{\mathrm{u}}=1$e-6, when UAV-to-UAV interference is negligible, whereas it decreases for $\lambda_{\mathrm{u}}=5$e-6, when UAV-to-UAV interference is more pronounced.

\begin{figure}
	\centering
	\includegraphics[width=\figwidth]{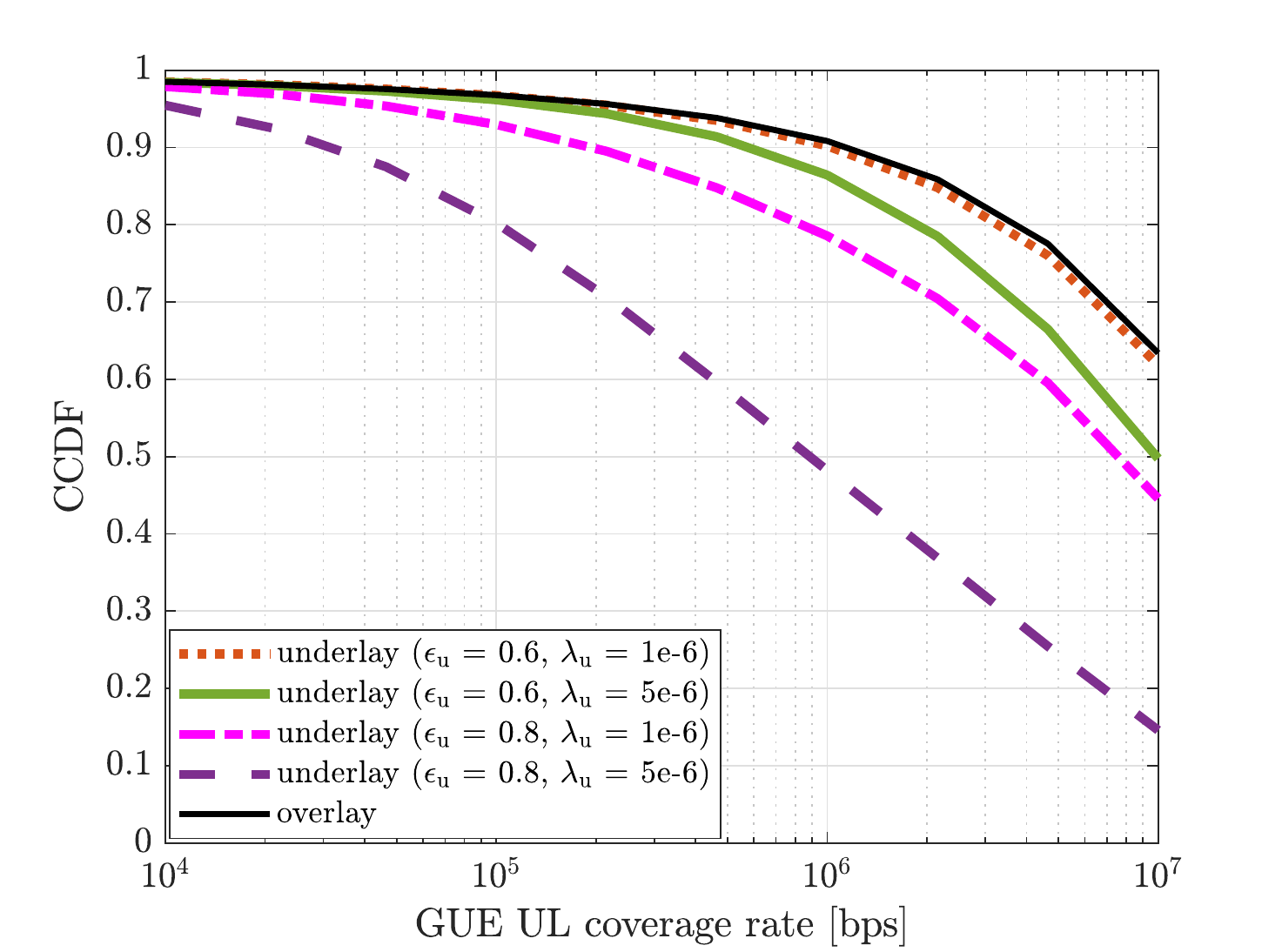}
	\caption{Coverage rate for GUE UL with underlay and overlay, for various values of $\epsilon_{\mathrm{u}}$ and $\lambda_{\mathrm{u}}$, with UAVs accessing five PRBs ($\eta_\u = 0.1$).}
	\label{CovRateGUE_combined}
\end{figure}

\begin{figure}
	\centering
	\includegraphics[width=\figwidth]{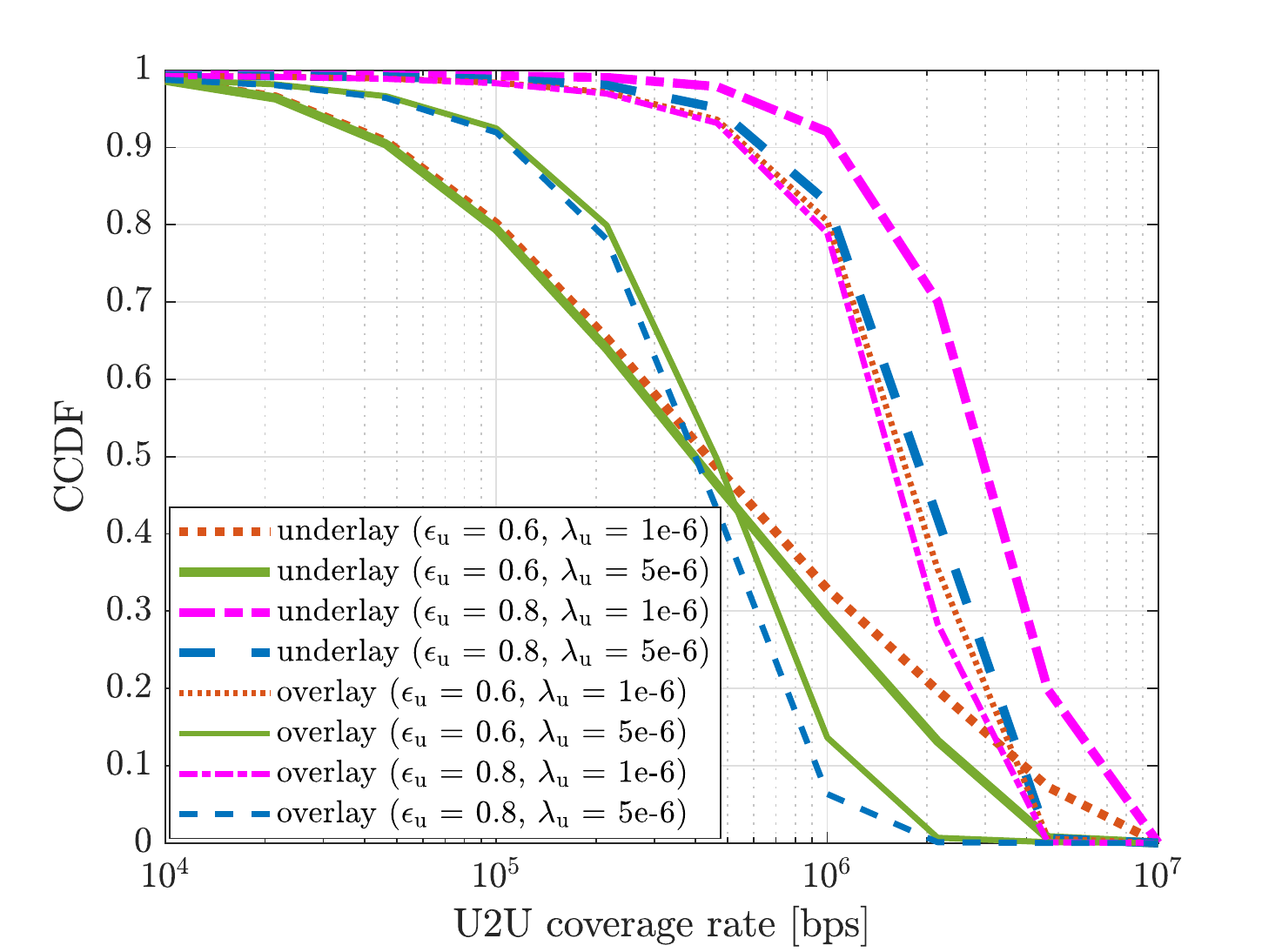}
	\caption{Coverage rate for U2U links with underlay and overlay, for various values of $\epsilon_{\mathrm{u}}$ and $\lambda_{\mathrm{u}}$, with UAVs accessing five PRBs ($\eta_\u = 0.1$).}
	\label{CovRateU2U_combined}
\end{figure}

Fig.~\ref{CovRateGUE_combined} and Fig.~\ref{CovRateU2U_combined} show the CCDF of the coverage rate for GUE UL and U2U links, respectively, when UAVs access five PRBs, in the underlay or in the overlay. 

Fig.~\ref{CovRateGUE_combined} demonstrates that in order to maintain a high GUE UL rate, one should i) adopt an overlay spectrum sharing approach, or ii) limit the power employed by the UAVs in the underlay, i.e., set $\epsilon_{\mathrm{u}}=0.6$. From a GUE perspective, the overlay approach seems the preferred alternative, since the effectiveness of applying a more conservative UAV power control policy diminishes when the UAV density grows (i.e., when $\lambda_{\mathrm{u}}=5$e-6).

Fig.~\ref{CovRateU2U_combined} provides the following insights:
\begin{itemize}[leftmargin=*]
\item In the overlay, the U2U coverage rate is only affected by UAV-to-UAV interference. Higher UAV densities thus have a more noticeable impact on the coverage rates than the UL power control strategy does. This can be observed by comparing scenarios with $\lambda_{\mathrm{u}}=1$e-6 (dotted red and dash-dotted purple curves) to scenarios with $\lambda_{\mathrm{u}}=5$e-6 (resp. solid green and dashed blue).
\item In the underlay, the U2U coverage rate is mostly affected by GUE-generated interference. Indeed, the rate degradation caused by increasing $\lambda_{\mathrm{u}}$ from 1e-6 to 5e-6 is limited when $\epsilon_{\mathrm{u}}=0.8$ (thick dash-dotted purple vs. dashed blue curves) and almost negligible when $\epsilon_{\mathrm{u}}=0.6$ (thick dotted red vs. solid green curves).
\item Comparing underlay vs. overlay, a crossover can be observed between green solid lines ($\epsilon_{\mathrm{u}}=0.6$, $\lambda_{\mathrm{u}}=5$e-6). This can be explained as follows. The upper part of the underlay CCDF corresponds to the worst U2U links---severely interfered by GUEs---which are better off in the overlay, where such interference is not present. The lower part of the underlay CCDF corresponds to the best U2U links---those not severely interfered by GUEs, for which UAV-to-UAV interference is dominant---that are worse off in the overlay, where all UAV interferers are concentrated on each PRB.
\item Comparing Fig.~\ref{CovRateGUE_combined} and Fig.~\ref{CovRateU2U_combined}, we observe that the overlay spectrum sharing approach is capable of i) offering the best guaranteed GUE UL performance, and ii) generally allowing a larger number of UAVs to achieve rates of 100 kbps---a requirement set by the 3GPP for command and control information exchange \cite{3GPP36777}.
\end{itemize}
\section{Conclusions}
\label{sec:conclusion}

In this article, we accurately evaluated the performance of an UL cellular network with both overlayed and underlayed U2U communications. We found that in the underlay, the U2U rate degradation caused by increasing the UAV density is limited, since the interference on U2U links is dominated by GUE transmissions. This turned out not to be true for the overlay, where higher UAV densities result in lower U2U rates, owing to all UAVs sharing the same resources without frequency hopping. Our results showed that, in an urban scenario, overlaying U2U and GUE UL communications is the preferable alternative  for simultaneously i) maximizing the GUE UL performance, and ii) meeting a minimum U2U coverage rate requirement of 100~kbps.

\ifCLASSOPTIONcaptionsoff
  \newpage
\fi
\bibliographystyle{IEEEtran}
\bibliography{Strings_Gio,Bib_Gio,Bib_Mahdi}
\end{document}